\documentclass[twoside,11pt,final]{entics} 
\usepackage{enticsmacro}
\usepackage{graphicx,amsmath,prooftree,stmaryrd}
\usepackage[all]{xy}
\newcommand{\Int}{{\mathbb Z}}

\newcommand{\Var}{\mathsf{Var}}

\newcommand{\Dom}{\mathsf{Dom}}
\newcommand{\Supp}{\mathsf{Supp}}

\newcommand{\Dist}{\mathcal{D}}
\newcommand{\State}{\mathsf{State}}

\newcommand{\pto}{\rightharpoonup}

\newcommand{\Defined}{\!\downarrow}

\newcommand{\Fail}{\mathsf{fault}}

\newcommand{\rMap}[1]{\xrightarrow{#1}}

\newcommand{\dTrans}{\rMap{~~~~}}
\newcommand{\lTrans}[1]{\rMap{\,~#1~\,}}
\newcommand{\pAssgn}[2]{#1\xleftarrow{\$}#2}
\newcommand{\pTrans}[2]{\rMap{\,\pAssgn{#1}{#2}\,}}

\newcommand{\Sem}[1]{\llbracket #1 \rrbracket}

\newcommand{\Skip}{\mathsf{skip}}
\newcommand{\If}{\mathsf{if}}
\newcommand{\Then}{\mathsf{then}}
\newcommand{\Else}{\mathsf{else}}
\newcommand{\While}{\mathsf{while}}
\newcommand{\Do}{\mathsf{do}}

\newcommand{\True}{\mathsf{tt}}
\newcommand{\False}{\mathsf{ff}}

\newcommand{\RState}[2]{{#1 \! \restriction_{#2}}}
\newcommand{\PState}[2]{#1^{#2}}
\newcommand{\RPath}[2]{{#1 \! \restriction_{#2}}}

\newcommand{\Prob}{\mathsf{Pr}}
\newcommand{\Final}{\mathsf{final}}

\newcommand{\TP}{\mathcal{TP}}
\newcommand{\IP}{\mathcal{IP}}

\newcommand{\sqleq}{\sqsubseteq}

\newcommand{\Epi}{\twoheadrightarrow}

\newcommand{\Indep}{\mathop{\bot\!\!\!\!\!\;\bot}}

\newcommand{\Imp}{\rightarrow}

\newcommand{\Tau}{\mathrm{T}}

\newcommand{\MV}{\mathrm{MV}}
\newcommand{\RV}{\mathrm{RV}}
\newcommand{\WV}{\mathrm{WV}}
\newcommand{\FV}{\mathrm{FV}}

\newcommand{\Det}{\mathsf{Det}}

\newcommand{\Mask}[2]{\langle #1 \rangle #2}

\newcommand{\CM}[1]{\mathbf{C}_{\!#1}}

\def\conf{MFPS 2025} 	
\volume{5}			
\def\volu{5}			
\def\papno{12}			
\def\mydoi{16743}
\def\lastname{Jereb \& Simpson} 
\def\runtitle{The Probabilistic Frame Rule} 
\def\copynames{J. Jereb, A. Simpson}    
\begin{document}
\begin{frontmatter}
  \title{Safety, Relative Tightness \\
  and the Probabilistic Frame Rule} 						
  \author{Janez Ignacij Jereb\thanksref{a}\thanksref{b}}	
   \author{Alex Simpson\thanksref{b}\thanksref{c}\thanksref{myemail}\thanksref{support}}		
   \address[a]{Faculty of Computer and Information Science (FRI)\\ University of Ljubljana, Ljubljana, Slovenia}  
   \address[b]{Faculty of Mathematics and Physics (FMF)\\ University of Ljubljana, Ljubljana, Slovenia}  			
   \address[c]{Institute for Mathematics, Physics and Mechanics (IMFM) \\ Lljubljana, Slovenia} 
  \thanks[myemail]{Email:  \href{mailto:Alex.Simpson@fmf.uni-lj.si} {\texttt{\normalshape
       Alex.Simpson@fmf.uni-lj.si}}}
    \thanks[support]{Research supported by ARIS research programme P1 0294.}
\begin{abstract} 
Probabilistic separation logic offers an approach to reasoning about imperative probabilistic programs in which a separating conjunction is used as a mechanism for expressing independence properties. Crucial to the effectiveness of the formalism is the frame rule, which enables modular reasoning about independent probabilistic state. We explore a semantic formulation of probabilistic separation logic, in which the frame rule has the same simple formulation as in separation logic, without further side conditions. This is achieved by building a notion of safety into specifications, using which we establish a crucial property of specifications, called relative tightness, from which the soundness of the frame rule follows.
\end{abstract}
\end{frontmatter}

\section{Introduction}
\label{section:introduction}

Separation logic is an immensely successful formalism for the verification of 
 imperative programs involving memory manipulation~\cite{OHRY,YOH}.
 Critical to its success is the \emph{frame rule}, which allows 
 program verification to be carried out in a modular style, in which the specification and verification of subroutines
 makes reference only  to resources  local to the subroutine. 
 
Probabilistic separation logic~\cite{BHL} adapts separation logic to programs 
with randomness. A main guiding aim is to utilise the power of the 
separating conjunction $*$ of separation logic and its associated frame rule to 
provide a modular means of verifying  properties of probabilistic programs. In this setting, the formula $\Phi * \Psi$
asserts that $\Phi$ and $\Psi$ hold in \emph{probabilistically independent} parts of memory.
The original paper~\cite{BHL} established the framework for a fragment of pwhile (a simple probabilistic imperative language), and
gave many examples of verification tasks (taken from cryptography) that can be handled by the approach. 
Subsequent developments have extended the framework with more elaborate probabilistic concepts such as
negative dependencies~\cite{BGHT}, and conditional independence~\cite{BDHS}.
An interesting related development is the Lilac program logic~\cite{LAS}, which performs a
similar task for probabilistic functional programs without mutable state.
One of the advantages of the functional framework is that it allows the frame rule to have a
simple, elegant formulation.

\begin{figure}
\begin{gather*}
\left(\,\FV(\Theta) \cap \MV(C) = \emptyset \quad 
\FV(\Psi) \subseteq T \cup \RV(C) \cup \WV(C)
\quad
\models \Phi \to \mathbf{D}[T \cup \RV(C)]
\,\right)
\\
\begin{prooftree}
\{\Phi\} \, C\, \{\Psi\}
\justifies
\,\{\Phi * \Theta \} \, C\, \{\Psi * \Theta\}\,
\using{\,}
\end{prooftree}
\end{gather*}
\caption{The frame rule from~\cite{BHL}}
\label{figure:original-frame}
\end{figure}

In the original imperative setting of~\cite{BHL,BGHT}, in contrast, the probabilistic frame rule  has the more cumbersome form in Figure~\ref{figure:original-frame}, with the three side conditions written in the top line of the rule. The first side condition, which appears also in the frame rule of 
ordinary (heap) separation logic~\cite{OHRY,YOH}, requires the program $C$ not to modify the values of the free variables of the assertion  $\Theta$, a condition necessary to guarantee that the truth of $\Theta$  is preserved by the execution of $C$.  The second and third side condition are specific to probabilistic separation logic. The second requires $\Psi$ to only involve free variables that are in the union of: a chosen set of variables $T$, the set $\RV(C)$ of variables that may be read by $C$ before being written to, and the set $\WV(C)$ of variables that $C$ must write to before they are read.
The third side condition requires $\Phi$ to guarantee that all variables in $T$ and $\RV(C)$ have defined values. 

In addition to the convoluted formulation of the frame rule, another feature of the papers~\cite{BHL,BGHT} is that  significant restrictions are imposed on the 
the programs  they consider. One restriction is  that a syntactic distinction is maintained between deterministic and probabilistic variables in programs, which leads to restrictions on the positioning of assignment statements, and which is too absolute to be able to cater for program reasoning in which one needs to take  account of the local deterministic behaviour of a globally probabilistic variable. Another restriction is that loop guards are required to be deterministic. Yet another is that loops are restricted to a bounded number of iterations.

Since imperative probabilistic languages with mutable state are an important
paradigm with many applications, and probabilistic separation logic is a promising approach to verification
for such languages, it is worth investigating if some of the complications and restrictions above can be avoided.
In this paper, we propose a framework for addressing this.  
Working  with an unrestricted programming language (the full pwhile language), our main focus will be on 
understanding the frame rule and simplifying its formulation. We hope that future work will
show that our framework offers  a basis  for extending the proof rules of~\cite{BHL,BGHT} to support program verification for the unrestricted language. 
 
Following the lead of heap separation logic~\cite{IOH}, probabilistic separation logic has 
hitherto been built on the elegant and versatile framework of the logic of bunched implications~\cite{OHP,Pym}
and its resource-monoid models~\cite{POHY}. In the probabilistic case~\cite{BHL,BGHT}, one 
obtains a satisfaction relation of the form
$\Sigma \models \Phi$, where $\Sigma$ is a random state and $\Phi$ is an assertion. 
(We are deliberately simplifying the formulation to ease the discussion.) 
In order to obtain the frame rule of the present paper, we found we needed
to take a small step back from this framework. While we also have a satisfaction `relation'
of the form $\Sigma \models \Phi$, we take the  truth value of the relation as being \emph{undefined}
in the case that the random state $\Sigma$ does not provide sufficient information to determine the truth or falsity of $\Phi$.
In the case that the relation $\Sigma \models \Phi$ is defined, it assumes either the value 
$\True$ (which means $\Sigma$ satisfies $\Phi$) of $\False$ (which means
$\Sigma$ satisfies $\neg \Phi$), and indeed the laws of classical logic are validated.
Thus we avoid
one of the idiosyncrasies  of the probabilistic separation logics in the literature, namely
that they are inherently intuitionistic. 

Another minor departure from the literature is our formulation of random state.
 There are two natural ways to approach this. The standard approach in the literature is to view 
 a random state as given by a  probability distribution on states. This approach goes back to~\cite{Kozen}.
In this paper, we instead
view a random state as a 
state-valued random variable, defined in the usual way as a function from a sample space to states.
We find that this approach allows
a more intuitive and mathematically convenient formalism for working with separating conjunction and its associated frame rule. Nevertheless, it should be emphasised that the two approaches are equivalent in the sense that one can translate from either one to the other. 

Within the framework of random state as a random variable, the separating conjunction $\Phi*\Psi$ has a very straightforward meaning. It holds in a random state $\Sigma$ just when there are sets $U,V$ of program variables
such that: $\Phi$ and $\Psi$ hold in the projections $\RState{\Sigma}{U}$ and $\RState{\Sigma}{V}$ of $\Sigma$ to $U$ and $V$ respectively, 
and additionally $\RState{\Sigma}{U}$ and
$\RState{\Sigma}{V}$ are independent random variables. Departing from the literature on imperative probabilistic separation logic, we do not
impose the requirement  that the sets $U$ and $V$ be disjoint.\footnote{The disjointness requirement is also not present in the functional setting of~\cite{LAS}.}
If  these sets  do overlap, then  the independence condition forces
all variables in the overlap to be deterministic. This allows us to accommodate the sharing of deterministic variables between $\Phi$ and $\Psi$, a feature of probabilistic separation logic which is important to applications, without needing a separate syntactic class of deterministic variables.

Our approach to probabilistic separation logic will be semantic. Rather than fixing a specific syntax for the logic, we define a semantic notion of assertion, and introduce connectives as
operations on semantic assertions. This approach,
which allows us to focus on semantic principles,
is very much influenced by~\cite{YOH}, which uses a similar approach to provide a semantic underpinning for the frame rule of heap separation logic. Our development is also greatly influenced by the conceptual analysis of the frame rule in \emph{op.\ cit.}, in which it is related to two properties of specifications: \emph{safety} and \emph{tightness}. Safety means that specifications $\{\Phi\} C \{\Psi\}$ include a guarantee that if $C$ is run
from a state satisfying the precondition $\Phi$ then the execution does not fault.
Tightness means that the specification mentions all resources involved in the execution of $C$. 
In our probabilistic framework, specifications $\{\Phi\} C \{\Psi\}$ will indeed include
a \emph{safety} component: 
if a random state $\Sigma$ satisfies $\Phi$, then $C$, run from $\Sigma$, 
does not incur a memory fault.
As a consequence of this safety property, we shall
derive a precisely formulated version of  tightness, which we call \emph{relative tightness}: 
$\Phi$ necessarily specifies everything about the random state $\Sigma$ that is relevant to the 
behaviour of $C$ on all resources needed for interpreting the postcondition $\Psi$. (Here ``relative'' means relative to $\Psi$.)
In our framework, in which random states are random variables, this 
relative tightness property has an elegant formulation in terms of conditional independence
(Theorem~\ref{theorem:independence}).
The soundness of the frame rule, which appears in its  original formulation without further side conditions, then follows
from relative tightness by a pleasingly abstract argument exploiting general properties of independence and conditional independence (Corollary~\ref{corollary:frame}).

After reviewing mathematical preliminaries in Section~\ref{section:preliminaries}, we present the pwhile language in Section~\ref{section:language} together with a small-step operational semantics. In Section~\ref{section:assertions},
we introduce our semantic notion of assertion based on random state. Section~\ref{section:specifications}
treats specifications, given as Hoare triples, and defines their partial-correctness and total-correctness interpretations, which, importantly, come with  a built-in safety guarantee.
The main contributions of the paper then appear in Section~\ref{section:frame}. We present the frame rule and prove its soundness utilising the notion of relative tightness, which is formulated in Theorem~\ref{theorem:independence}. We also provide a simple counterexample showing that 
relative tightness and the soundness of the frame rule both fail if the safety guarantee of specifications is dropped.
Finally, in Section~\ref{section:proof}, 
we discuss prospects for extending the setting in this paper to a fully-fledged verification logic for the pwhile language.

\section{Mathematical preliminaries}
\label{section:preliminaries}

A \emph{(discrete probability) distribution} on a set $A$ is a function $d \colon A \to [0,1]$ that satisfies 
\[
\sum_{a \in A} d(a) ~ = ~ 1 \, .
\]
Its \emph{support} is defined by
\[
\Supp(d) ~ := ~ \{a \in A \mid d(a) > 0\} \, .
\]
The support $\Supp(d)$ is necessarily a countable set. We write $\Dist(A)$ for the set of all distributions
on $A$. A function $f: A \to B$ induces a function $f_!: \Dist(A) \to \Dist(B)$ that maps every
$d \in \Dist(A)$ to its \emph{pushforward}
\[
f_!(d)(b) ~ := ~ \sum_{a \in f^{-1}(b)} d(a) \, .
\]

For an arbitrary set $A$, an $A$-valued \emph{random variable} is a function $\Sigma : \Omega \to \State$, where $\Omega$ is a (discrete)
\emph{sample space}, namely a (necessarily countable) set equipped with a probability distribution $p_\Omega \colon \Omega \to [0,1]$, for  which we require the additional (not strictly necessary, but natural)  condition that $\Supp(d) = \Omega$; that is, there are no 
zero-probability sample-space elements. 
It will sometimes be useful to change sample space along a \emph{morphism of sample spaces},
which  is a function $q : \Omega' \to \Omega$ satisfying, $q_!(p_{\Omega'}) = p_\Omega$.
Since sample spaces have no zero-probability elements, every morphism of sample spaces is surjective. To emphasise this point in the sequel, we shall write $q : \Omega' \Epi \Omega$ for sample space morphisms.

Given a random variable $S : \Omega \to A$, we write $p_S: A \to [0,1]$ for the \emph{distribution} of $S$,
given by the pushforward probability distribution $S_!(p_\Omega)$ on $A$. We define the \emph{support} of a random variable $S$ by $\Supp(S) := \Supp(p_S)$.

Two random variables $S : \Omega \to A$ and $T : \Omega \to B$ are said to be \emph{independent} (notation $S \Indep T$) if, for all $a \in A$ and $b \in B$, it holds that $p_{(S,T)}(a,b) = p_S(a) \cdot p_T(b)$, where the left-hand side refers to the induced random variable $(S,T) : \Omega \to A \times B$. It is obvious from this definition that independence is a symmetric relation. We shall need the following standard property of independence.
\begin{equation}
\label{indep-function}
\text{If $S \Indep T$ then, for any $f: B \to C$, it holds that $S \Indep f \circ T$}\, .
\end{equation}

Given two random variables $S : \Omega \to A$ and $T : \Omega \to B$, and given $a \in \Supp(S)$,
the \emph{conditioned random variable} $T_{|S=a}: ~ \Omega_{|S=a} \to B$ is defined by
\begin{align*}
\Omega_{|S=a} ~ & := ~ \{\omega \in \Omega \mid S(\omega) = a \} \\
p_{\Omega_{|S=a}}(\omega)  ~ & := ~ \frac{p_\Omega(\omega)}{p_{S}(a)} \\[1ex]
T_{|S=a} (\omega) ~ & := ~ T(\omega) \, .
\end{align*}
Independence can be characterised in terms of conditioning, \emph{viz}
\begin{equation}
\label{independence-via-conditioning}
S \Indep T ~ \iff ~ \forall b, b' \in \Supp(T), ~ p_{S_{|T=b}} = p_{S_{|T=b'}}\, .
\end{equation}
The relation of \emph{conditional independence} of $S : \Omega \to A$ and $T : \Omega \to B$ 
\emph{given} $U : \Omega \to C$  is defined by 
\begin{equation}
\label{cond-indep-def}
S \Indep T \mid U ~~ {:\!\!\iff} ~ ~
\forall c \in \Supp(U), ~S_{|U=c} \Indep T_{|U=c} \,. 
\end{equation}
We shall make use of the following standard property of conditional independence.
\begin{equation}
\label{cond-indep-rule}
\text{If $S \Indep T \mid U$ and $S \Indep U$ then $S \Indep (T,U)$}\, .
\end{equation}

The paper will study a 
programming language pwhile, introduced in Section~\ref{section:language}, which manipulates state. For us, a \emph{state} $\sigma$ is a finite partial function $\sigma : \Var \pto \Int$,
where $\Var$ is a fixed set of program variables. Equivalently, a state is 
a total function from a finite set   $\Dom(\sigma) \subseteq \Var$, the \emph{domain} of $\sigma$, to $\Int$.
We write $\State$ for the set of states, which carries a natural partial order 
\begin{equation}
\label{eqn:state-order}
\sigma \sqleq \sigma' ~ :\Longleftrightarrow ~ \forall X \in \Var,\, \sigma(X)\Defined \implies \sigma'(X)= \sigma(X)\, ,
\end{equation}
where we write $\sigma(X)\Defined$ to mean that $X \in \Dom(\sigma)$.
A useful operation will be to \emph{restrict} a state $\sigma$ to a given finite  $V \subseteq \Var$, resulting in the state
\begin{equation}
\label{eqn:state-restriction}
\RState{\sigma}{V}(X) ~ := ~ 
\begin{cases}
\sigma(X) & \text{if $X \in V$} \\
\bot & \text{otherwise} \, ,
\end{cases}
\end{equation}
where we use $\bot$ to denote undefinedness. Note that $\RState{\sigma}{V}(X)$ is undefined for those $X \in V$ for which $\sigma(X)$ is undefined. In the special case that $\sigma(X)$ is defined for all $X \in V$ (equivalently if $\Dom(\sigma) \subseteq V$), we say that $\sigma$ is \emph{$V$-total}.

In order to model the probabilistic assertion logic of Section~\ref{section:assertions}, we need to randomise state.
A \emph{random state} is simply a $\State$-valued random variable $\Sigma : \Omega \to \State$. 
The partial order on $\State$~\eqref{eqn:state-order}
induces an associated pointwise partial order on the set of $\State$-valued random variables with common sample space $\Omega$. Suppose $\Sigma,\Sigma' : \Omega \to \State$, then
\begin{equation}
\label{eqn:random-state-order}
\Sigma \sqleq \Sigma' ~ :\Longleftrightarrow ~ \forall \omega \in \Omega,\, \Sigma(\omega) \sqleq \Sigma'(\omega) \, .
\end{equation}

The restriction 
operation on states~\eqref{eqn:state-restriction} 
induces an associated operation on random state. 
Given a random state $\Sigma: \Omega \to \State$ and finite $V \subseteq \Var$, the 
restricted random state
$\RState{\Sigma}{V}: \Omega \to \State$ is defined by:
\[ \RState{\Sigma}{V}(\omega) ~ := ~ \RState{\Sigma(\omega)}{V}\, . \]
The random state $\Sigma$ is said to be \emph{$V$-total} if $\Prob[\text{$\Sigma$ is $V$-total}] = 1$.
 In spite of appearances, this is not a circular definition. We are of course using standard notational conventions for random variables to abbreviate the property
 \[\sum_{\omega \in \Omega}  \begin{cases} 
p_\Omega(\omega) & \text{$\Sigma(\omega)$ is $V$-total} \\
0 & \text{otherwise.}
\end{cases} 
~ = ~ 1 \, .\]
In the case that $\Sigma$ is $\{X\}$-total, we write $\PState{\Sigma}{X} : \Omega \to \Int$ for the random integer
\begin{equation}
\label{eqn:Sigma-X}
\PState{\Sigma}{X}( \omega) ~ :=  ~ \Sigma(\omega)(X) \, .
\end{equation}

\section{Language}
\label{section:language}


We introduce a version of the simple probabilistic while language pwhile. We use $X,Y,Z\dots$ to range over
an infinite set $\Var$ of integer variables. We use $E, \dots$ to range over integer expressions and $B, \dots$ to range over boolean expressions. Such expressions are built as usual using  variables, basic arithmetic operations, basic relations and logical connectives. We leave the precise syntax open. In the programming language, we also include \emph{distribution expressions}, which denote probability distributions over the integers. 
Distribution expressions may contain integer subexpressions, allowing us to define distributions parametrically.
For example, we might have an expression $\mathsf{uniform}(X)$, denoting the uniform probability distribution on
the integer interval $[\min(0,X),\max(0,X)]$. We use $D, \dots$ to range over distribution expressions. Again, we leave the precise syntax unspecified. With the above conventions, the syntax of commands is given by
\begin{align*}
C ~~ ::= ~~ & X := E \mid \pAssgn{X}{D} \mid \Skip \mid C ; C \mid \If\ B\ \Then\ C\ \Else\ C \mid \While\ B\ \Do\ C \, ,
\end{align*}
where the command $\pAssgn{X}{D}$ randomly samples from the distribution $D$ and assigns the returned value to the variable $X$.

We define a small-step operational semantics, by specifying a probabilistic transition system, with mid-execution transitions between configurations of the form $C, \sigma$, where $\sigma \in \State$ as defined in
Section~\ref{section:preliminaries} Execution may terminate in a terminal configuration $\sigma$.
It may also abort, due to a memory fault, in the  $\Fail$ error node.
The transition system thus has the following nodes.
\begin{itemize}
\item \emph{Nonterminal configurations:} $C,\sigma$, where $C$ is a command and $\sigma$ a state.
\item \emph{Terminal configurations:} $\sigma$, where $\sigma$ is a state.
\item \emph{The error node:}  $\Fail$.
\end{itemize}
There are two types of transition.
\begin{itemize}
\item \emph{Deterministic transitions:}  unlabelled transitions  $u \dTrans v$, where $u$ is nonterminal.
\item \emph{Probabilistic transitions:} labelled transitions $u \pTrans{n}{d} v$, where $u$ is nonterminal, $d$ is a probability distribution on $\Int$ and $n \in \Supp(d)$.
\end{itemize}
The single-step transitions between nodes are determined inductively using the rules in Figure~\ref{fig:op-sem}.
In these rules, we write $\Var(E)$ for the set of variables appearing in an integer expression $E$, and we use similar notation for boolean expressions $B$ and distribution expressions $D$. 
The semantic interpretation $\Sem{E}_\sigma$ of an integer expression $E$, relative to a state $\sigma$, 
is defined if and only if
$\Var(E) \subseteq \Dom(\sigma) $. Similarly, for boolean  and distribution expressions, $\Sem{B}_\sigma$ and
$\Sem{D}_\sigma$
are defined
if and only 
$\Var(B) \subseteq \Dom(\sigma)$ and $ \Var(D) \subseteq \Dom(\sigma)$ respectively. When defined, we have  $\Sem{E}_\sigma \in \Int$, 
$\,\Sem{B}_\sigma \in \{\True,\False\}$ (our notation for truth values) and $\Sem{D}_\sigma \in \Dist(\Int)$.
Whenever  an expression is evaluated in a state in which some of its variables do not have a value, a memory fault is triggered. In the two rules involving variable assignment, $\sigma[X \mapsto n]$ is used for the state that 
maps $X$ to $n$, and agrees with $\sigma$ (which may or may not contain $X$ in its domain) otherwise.The other notational convention in
Figure~\ref{fig:op-sem} is that $L$ is a meta-variable for transition labels, which may be empty in the case of deterministic transitions, or of the form $\pAssgn{n}{d}$ in the case of probabilistic transitions.

\begin{figure}
\begin{gather*}
\begin{prooftree}
\justifies
X := E\, , \, \sigma  \dTrans \sigma [X \mapsto n]
\using{\Sem{E}_\sigma = n}
\end{prooftree}
\qquad
\begin{prooftree}
\justifies
X := E\, , \, \sigma  \dTrans \Fail
\using{\Var(E) \not\subseteq \Dom(\sigma)}
\end{prooftree}
\\[2ex]
\begin{prooftree}
\justifies
\pAssgn{X}{D}\, , \, \sigma  \pTrans{n}{d} \sigma [X \mapsto n]
\using{\Sem{D}_\sigma = d\, ,\, n \in \Supp(d)}
\end{prooftree}
\qquad
\begin{prooftree}
\justifies
\pAssgn{X}{D}\, , \, \sigma  \dTrans \Fail
\using{\Var(D) \not\subseteq \Dom(\sigma)}
\end{prooftree}
\\[2ex]
\begin{prooftree}
\justifies
\Skip,\sigma  \dTrans \sigma
\end{prooftree}
\\[2ex]
\begin{prooftree}
C_1, \sigma \lTrans{L} C_1', \sigma'
\justifies 
C_1;C_2\, , \, \sigma \lTrans{L} C_1';C_2 \, , \,\sigma'
\end{prooftree}
\qquad
\begin{prooftree}
C_1, \sigma \lTrans{L} \sigma'
\justifies 
C_1;C_2\, , \, \sigma \lTrans{L} C_2 \, , \,\sigma'
\end{prooftree}
\qquad
\begin{prooftree}
C_1, \sigma \dTrans \Fail
\justifies 
C_1;C_2\, , \, \sigma \dTrans \Fail
\end{prooftree}
\\[2ex]
\begin{prooftree}
\justifies 
\If\ B\ \Then\ C_1\ \Else\ C_2\, , \, \sigma \dTrans C_1, \sigma
\using{\Sem{B}_\sigma = \True}
\end{prooftree}
\qquad
\begin{prooftree}
\justifies 
\If\ B\ \Then\ C_1\ \Else\ C_2\, , \, \sigma \dTrans C_2, \sigma
\using{\Sem{b}_\sigma = \False}
\end{prooftree}
\\[2ex]
\begin{prooftree}
\justifies 
\If\ B\ \Then\ C_1\ \Else\ C_2\, , \, \sigma \dTrans \Fail
\using{\Var(b) \not\subseteq \Dom(\sigma)}
\end{prooftree}
\\[2ex]
\begin{prooftree}
\justifies 
\While\ B\ \Do\ C \, , \, \sigma \dTrans \sigma
\using{\Sem{B}_\sigma = \False}
\end{prooftree}
\qquad
\begin{prooftree}
\justifies 
\While\ B\ \Do\ C\, , \, \sigma \dTrans C\, ; \While\ B\ \Do\ C \, , \, \sigma
\using{\Sem{B}_\sigma = \True}
\end{prooftree}
\\[2ex]
\begin{prooftree}
\justifies 
\While\ B\ \Do\ C\, , \, \sigma \dTrans \Fail
\using{\Var(b) \not\subseteq \Dom(\sigma)}
\end{prooftree}
\end{gather*}
\caption{Operational Semantics}
\label{fig:op-sem}
\end{figure}

Every transition path from a nonterminal configuration $C_0, \sigma_0$ to a terminal configuration  has the form 
\begin{equation}
\label{eqn:terminal-path}
C_0,\sigma_0 \lTrans{L_0} C_1,\sigma_1  \lTrans{L_1}  \cdots\ C_{n-1},\sigma_{n-1} \lTrans{L_{n-1}} \tau \, ,
\end{equation}
for some $n \geq 1$, where each label $L_i$ is either empty or of the form $\pAssgn{n_i}{d_i}$. We call such 
paths \emph{terminal paths}.  The  \emph{probability} and the \emph{final state} of a terminal path $\pi$, of the form~\eqref{eqn:terminal-path}, are defined by
\begin{align}
\label{prob-pi}
\Prob(\pi) ~ := ~  & \prod_{i < n}  \begin{cases} d_i(n_i) & \text{if $L_i$ is $\pAssgn{n_i}{d_i}$} \\
1 & \text{if $L_i$ is empty\, ,}
\end{cases}
\\
\nonumber
\Final(\pi) ~ := ~  & \tau \, .
\end{align}
We write $\TP_{C, \sigma}$ for the set of all terminal paths from initial configuration $C,\sigma$.
The \emph{termination probability} for $C, \sigma$ is defined by
\[
\Prob[\text{$C,\sigma$ terminates}] ~ := ~ \sum_{\pi \, \in\, \TP_{C,\sigma}} \Prob(\pi)\, .
\]
We say that $C,\sigma$ is \emph{terminating} if $\Prob[\text{$C,\sigma$ terminates}] = 1$. This is the form of termination known as almost sure termination. We shall not concern ourselves with other forms of termination. A terminating configuration $C,\sigma$ determines a probability distribution on result states 
$p_{C,\Sigma} : \State \to [0,1]$, defined by:
\begin{equation}
\label{terminal-dist}
p_{C,\Sigma} (\tau) ~ := ~ \sum_{\pi \, \in\, \TP_{C,\sigma}} \Prob(\pi) \cdot \delta_{\tau,\, \Final(\pi)} \enspace ,
\end{equation}
where $\delta$ is the Kr\"onecker delta.

We consider \emph{divergence} to be the computational behaviour that follows an infinite path  of nonterminal configurations 
\begin{equation}
\label{eqn:infinite-path}
C_0,\sigma_0 \lTrans{L_0} C_1,\sigma_1  \lTrans{L_1}  \ C_{1},\sigma_{1} \lTrans{L_2} \cdots .
\end{equation}
We write $\IP_{C_0, \sigma_0}$ for the set of all such infinite paths from initial configuration $C_0,\sigma_0$.
We write $\RPath{\IP_{C_0, \sigma_0}}{n}$ for the set of $n$-length prefixes of paths in $\IP_{C_0, \sigma_0}$,
where the length is the number of transitions.
The \emph{divergence probability} for $C_0, \sigma_0$ is defined by
\begin{equation}
\label{eqn:divergence-prob}
\Prob[\text{$C_0,\sigma_0$ diverges}] ~ := ~ 
\lim_{n \to \infty} \sum_{\gamma\, \in \,\RPath{\IP_{C_0, \sigma_0}}{n}} \Prob(\gamma) \, ,
\end{equation}
where $\Prob(\gamma)$ is calculated analogously to~\eqref{prob-pi} above.
It is always the case that
\begin{equation}
\label{terminate-inequality}
\Prob[\text{$C_0,\sigma_0$ terminates}] + \Prob[\text{$C_0,\sigma_0$ diverges}] ~ \leq ~ 1 \, .
\end{equation}
The inequality is an equality precisely in the case that the configuration $C_0,\sigma_0$ is 
\emph{fault-free}, meaning that  it is 
not the case that $C,\sigma \dTrans^* \Fail$, where 
we write $\dTrans^*$ for the transitive-reflexive closure of the single-step transition relation, including both labelled and unlabelled transitions in the relation. Note that terminating configurations are \emph{a fortiori} fault-free.


\section{Assertions}
\label{section:assertions}

Rather than focusing on a specific  assertion logic, we work with a general semantic notion of assertion. To motivate this, we
give some examples of assertions that fit into the framework. A first example is
\[
X \sim d \, ,
\]
where $d$ is some probability distribution on $\Int$. This says that the value of the program variable $X$ is distributed probabilistically according to $d$. Another simple example is
\[
[X = Y] \, .
\]
Since the values of variables are in general randomly distributed, we interpret the above equality as stating that $X$ and $Y$ take the same value with probability $1$; that is, considered as random variables, $X$ and $Y$ are {almost surely equal}.

In order to accommodate randomness in the semantics of assertions, 
our main satisfaction relation will have the form $\Sigma \models \Phi$, relating random states $\Sigma$ and 
semantic assertions $\Phi$. However, we add one further ingredient to the picture. 
We allow the truth value of the `relation' 
$\Sigma \models \Phi$ to be undefined for certain ineligible random states $\Sigma$. This possibility of undefinedness is included as  a natural mathematical
mechanism  for  dealing with the partiality of state. For example, we consider $\Sigma \models [X = Y]$ to be defined if and only if the random state $\Sigma$ assigns values to both $X$ and $Y$ with probability one; that is,
if $\Sigma$ is $\{X,Y\}$-total.
We write $(\Sigma \models \Phi)\Defined$ to mean that the value of the satisfaction relation 
$\Sigma \models \Phi$ is defined, in which case we either have
$(\Sigma \models \Phi) = \True$ or $(\Sigma \models \Phi) = \False$.
Again, in the case that $\Phi$ is the assertion $[X = Y]$, we have 
 that 
$(\Sigma \models [X=Y])\Defined$ if and only  if
$\PState{\Sigma}{X}, \PState{\Sigma}{Y} : \Omega \to \Int$, as in~\eqref{eqn:Sigma-X}, are well-defined. 
Then $(\Sigma \models [X=Y]) = \True$ if and only if $\Prob[\Sigma_X = \Sigma_Y] = 1$, or equivalently,
because there are no zero-probability sample-space elements, 
if the functions $\PState{\Sigma}{X}$ and $\PState{\Sigma}{Y}$ are equal. Similarly, $(\Sigma \models [X=Y]) = \False$ if and only if $\Prob[\PState{\Sigma}{X} = \PState{\Sigma}{Y}] < 1$, or equivalently 
$\PState{\Sigma}{X} \neq \PState{\Sigma}{Y}$. This and many other similar examples illustrate that allowing undefinedness enables us to make a natural distinction between the property expressed by $\Phi$ not making sense in the context of $\Sigma$, in which case $\Sigma \models \Phi$ is undefined,
and the case in which the property does make sense, in which case  $\Sigma \models \Phi$
possesses an actual truth value. 

We now precisely define our semantic notion of  assertion. A \emph{semantic assertion} $\Phi$ is given by a \emph{partial} function from random states to the set of truth values  $\{\True,\False\}$.
We write this function as 
\[ \Sigma ~ \mapsto ~ (\Sigma \models \Phi) \, . \]
The function is required to satisfy three conditions. 
\begin{description}
\item[(SA1)]
$\Sigma \sqleq \Sigma'$ and  $(\Sigma \models \Phi)\Defined$ implies 
$(\Sigma \models \Phi) = (\Sigma' \models \Phi)$.
(The $\sqleq$ relation is defined in~\eqref{eqn:random-state-order}.)


\item[(SA2)] $\Phi$ has an associated finite set $\FV(\Phi) \subseteq \Var$, its \emph{footprint variables},%
\footnote{When assertions are given syntactically as formulas, $\FV(\Phi)$  can be understood as the 
\emph{free variables} of the formula $\Phi$. For semantic assertions, we prefer the more neutral terminology \emph{footprint variables}, which is motivated by thinking of $\FV(\Phi)$ as the memory footprint associated with $\Phi$.}
satisfying
\[
(\Sigma \models \Phi) \Defined ~~  \iff ~~ 
\text{$\Sigma$ is $\FV(\Phi)$-total.}
\]

\item[(SA3)] If $q : \Omega' \Epi \Omega$ is a morphism of sample spaces then
$(\Sigma \circ q \models \Phi) = (\Sigma \models \Phi)$.

\end{description}
One consequence of (SA2) is that it prevents the everywhere undefined function
$(\Sigma \mapsto \bot)$ from being a semantic assertion, which is no loss. 
Secondly, (SA2) determines $\FV(\Phi)$ uniquely. Indeed, suppose
we have two finite $U,V \subseteq \Var$ such that
$(\Sigma \models \Phi) \Defined$ iff $\Sigma$ is $U$-total iff $\Sigma$ is $V$-total.
Let $\Sigma$ be some random state such that $(\Sigma \models \Phi) \Defined$. (Such a state exists, by the 
first observation about (SA2) above.) Since $\Sigma$ is $U$-total, so is $\RState{\Sigma}{U}$, and hence
$(\RState{\Sigma}{U} \models \Phi) \Defined$. This means that $\RState{\Sigma}{U}$ is $V$-total, which entails 
$V \subseteq U$.  A symmetric argument establishes $U \subseteq V$. Hence $U = V$. 

In combination, (SA1) and (SA2) have the following simple consequence that we shall use frequently.
\begin{equation}
\label{eqn:SAA-SAB}
(\Sigma \models \Phi) \Defined ~~ \implies ~~
{(\RState{\Sigma}{\FV(\Phi)} \models \Phi)\Defined}~ \text{and}~  (\Sigma \models \Phi) = (\RState{\Sigma}{\FV(\Phi)} \models \Phi) \, .
\end{equation}
We remark also that (SA3) is equivalent to saying that, for any assertion $\Phi$,
the value of $\Sigma \models \Phi$ depends only on the probability distribution $p_\Sigma$. Thus
the logic could perfectly well be given an equivalent formulation in terms of distributions on $\State$ rather than random state, as has been the norm hitherto in the literature on probabilistic separation logic~\cite{BHL,BGHT}.

We give several examples of semantic assertions and constructions on them, to illustrate the broad scope of the notion. 
\begin{itemize}

\item If $B$ is a boolean expression then the assertion $[B]$ is defined by
\begin{gather*}
\FV([B]) ~  := ~ \Var(B)  \\
\Sigma \models {[B]} ~ \, {:\!\iff} ~ \Prob[\Sem{B}_\Sigma = \True] = 1\, .
\end{gather*}
This definition illustrates the general style we shall follow below. The first clause specifies the set of footprint variables, which determines when the value of $(\Sigma \models [B])$ is defined. The second clause,
is then only considered in the case that $(\Sigma \models [B])\Defined\,$. 

\item A special case of the above is the assertion $[E_1 = E_2]$, for integer expressions $E_1,E_2$.

\item If $E$ and $D$ are integer and distribution expressions respectively then the assertion $E \sim D$ is defined by
\begin{gather*}
\FV(E \sim D) ~  := ~ \Var(E) \cup \Var(D) \\
\Sigma \models {E\sim D} ~ \, {:\!\iff} ~ 
\forall d \in \Supp(\Sem{D}_\Sigma), ~ p_{\,{(\Sem{E}_\Sigma)}_{|\Sem{D}_\Sigma = d}} = d \, .
\end{gather*}
Here, the expression $\Sem{D}_\Sigma$ is considered as the distribution-valued random variable $\omega \mapsto \Sem{D}_{\Sigma(\omega)}$, and the
property asserts that the distribution of $E$, conditional on any possible distribution $d$ arising from $D$, is itself $d$.

\item If $E$ is an integer expression then the assertion $\Det(E)$ is defined by
\begin{gather*}
\FV(\Det(E)) ~  := ~ \Var(E) \\
\Sigma \models {\Det(E)} ~ \, {:\!\iff} ~ 
\exists n \in \Int, ~ \Prob[\Sem{E}_\Sigma = n] = 1 \, . 
\end{gather*}
Informally: $E$ is deterministic.

\item If $\Phi$ is a semantic assertion then so is $\neg \Phi$ defined classically by
\begin{gather*}
\FV(\neg \Phi) ~  := ~ \FV(\Phi) \\
\Sigma \models {\neg \Phi} ~ \, {:\!\iff} ~ \Sigma \not \models \Phi
 \, . 
\end{gather*}
We emphasise that the second line only applies in the case that
$(\Sigma \models {\neg \Phi})\Defined$.

\item If $\Phi, \Psi$ are semantic assertions then 
$\Phi \vee \Psi$, $~\Phi \wedge \Psi$ and $\Phi \Imp \Psi$,
are defined in the standard classical way, for example
\begin{gather*}
\FV(\Phi \Imp \Psi) ~  := ~ \FV(\Phi) \cup \FV(\Psi)\\
\Sigma \models {\Phi \Imp \Psi} ~ \, {:\!\iff} ~ \text{$\Sigma \not\models \Phi$ or $\Sigma \models \Psi$,}
 \, . 
\end{gather*}

\item If $E$ is an integer expression and $\Phi$ is a semantic assertion then
so is $\CM{E} \,\Phi$ defined by
\begin{gather*}
\FV(\CM{E}\, \Phi) ~  := ~ \Var(E) \cup \FV(\Phi) \\
\Sigma \models {\CM{E}\, \Phi} ~ \, {:\!\iff} ~ \forall n \in \Supp(\Sem{E}_{\Sigma}), ~ 
\Sigma_{|\Sem{E}_{\Sigma} = n} \models \Phi \, . 
\end{gather*}
Here $\CM{E}$ is essentially the \emph{conditioning modality} introduced in~\cite{LAS}, which we model by conditioning the random state~$\Sigma$, using the conditioning operation on random variables from
Section~\ref{section:preliminaries}.

\item
If $\Phi, \Psi$ are semantic assertions then so is $\Phi * \Psi$ defined by
\begin{gather*}
\FV(\Phi * \Psi) ~  := ~ \FV(\Phi) \cup \FV(\Psi)\\
\Sigma \models \Phi * \Psi ~  {:\! \iff} ~ 
\Sigma \models \Phi\ \text{and}\ \Sigma \models \Psi\ \text{and}\ \RState{\Sigma}{\FV(\Phi)} \Indep \RState{\Sigma}{\FV(\Psi)} \, ,
\end{gather*}
\end{itemize}

\noindent
The last assertion above is of course the separating conjunction, which will play a critical role in the frame rule. One might expect the  second clause in its definition to have a different form; for example,
\begin{equation}
\label{star-B}
\RState{\Sigma}{\FV(\Phi)} \models \Phi\ \text{and}\ \RState{\Sigma}{\FV(\Phi)} \models \Psi\ \text{and}\ \RState{\Sigma}{\FV(\Phi)} \Indep \RState{\Sigma}{\FV(\Psi)} \, ;
\end{equation}
but this is equivalent by~\eqref{eqn:SAA-SAB}. 
Another equivalent is given by 
\begin{equation}
\label{star-C}
\exists U, V \subseteq \Var, ~
\RState{\Sigma}{U} \models \Phi\ \text{and}\ \RState{\Sigma}{V} \models \Psi\ \text{and}\ 
\RState{\Sigma}{U} \Indep \RState{\Sigma}{V} \, .
\end{equation}
Indeed, it is immediate that~\eqref{star-B} implies~\eqref{star-C}.
For the converse, suppose we have $U,V$ as in~\eqref{star-C}.
Since $\RState{\Sigma}{U} \models \Phi$ and $\RState{\Sigma}{V} \models \Psi$, we have,
by (SA2), that $\FV(\Phi) \subseteq U$ and $\FV(\Psi) \subseteq V$. So
$\RState{\Sigma}{\FV(\Phi)} \models \Phi$, by~\eqref{eqn:SAA-SAB},
because $\RState{\Sigma}{\FV(\Phi)} = \RState{\RState{\Sigma}{U\!}}{\FV(\Phi)}$.
Similarly, $\RState{\Sigma}{\FV(\Phi)} \models \Psi$.
Finally, since, if we compose $\sigma \mapsto \RState{\sigma}{\FV{\Phi}}$
with the left-hand side and $\sigma \mapsto \RState{\sigma}{\FV{\Psi}}$ with the right-hand side
 of the independence statement $\RState{\Sigma}{U} \Indep \RState{\Sigma}{V}$, then
it follows from~\eqref{indep-function} (and its symmetric version) that
$\Sigma_{\FV(\Phi)} \Indep \Sigma_{\FV(\Psi)}$.

As a final comment about $*$, we remark that it is possible for 
$\Sigma \models \Phi * \Psi$ to hold in cases in which $\FV(\Phi) \cap \FV(\Psi) \neq \emptyset$.
In such cases, by composing both sides of $\RState{\Sigma}{\FV(\Phi)} \Indep \RState{\Sigma}{\FV(\Psi)}$
with the function $\sigma \mapsto \RState{\sigma}{V}$, where $V := \FV(\Phi) \cap \FV(\Psi)$, it
follows from~\eqref{indep-function} that $\RState{\Sigma}{V} \Indep \RState{\Sigma}{V}$.
This means that $\RState{\Sigma}{V}$ is a deterministic random variable. That is, there exists
a $V$-defined state $\sigma$ such that $\Prob[\RState{\Sigma}{V} = \sigma] = 1$.
Thus our formulation of the separating conjunction, automatically entails the property that the two sides $\Phi$ and $\Psi$ can share deterministic variables, which in the original probabilistic separation logics is built-in 
via a syntactic distinction between deterministic and probabilistic variables~\cite{BHL,BGHT}.

\section{Specifications}
\label{section:specifications}

%
%
A \emph{specification} is given by a Hoare triple
\[
\{\Phi\} \, C \, \{\Psi\} \, ,
\]
where $\Phi$ and $\Psi$ are semantic assertions, and $C$ is a command from the programming language of Section~\ref{section:language}.
As is standard, we shall have two forms of correctness for specifications, \emph{partial} and \emph{total}, whose precise form, in which \emph{safety} is built in, is influenced by heap separation logic~\cite{OHRY,YOH}. Partial correctness will say that if $C$ is run in any random state $\Sigma$ for which $\Sigma \models \Phi$, then the execution of $C$ is fault-free (this is the safety component) and, if $C$ almost surely terminates, then 
$\Tau \models \Psi$, where $\Tau$ is the induced random state in which $C$ terminates. Total correctness more simply says that if $C$ is run in any random state $\Sigma$, for which $\Sigma \models \Phi$, then $C$ almost surely terminates (hence is \emph{a fortiori} fault-free), and $\Tau \models \Psi$, where $\Tau$ is the termination random state. To make this precise, we have to define the random state in which a terminating program terminates.

Let $\Sigma : \Omega \to \State$ be a random state. We view the pair $C, \Sigma$ as a random configuration (with deterministic first component). The notions of fault-freeness and termination for 
deterministic configurations, defined in Section~\ref{section:language}, extend to random configurations 
in the natural way. Specifically, 
we say that $C,\Sigma$ is \emph{fault-free} if
$\Prob[\text{$C, \Sigma$ is fault-free}] = 1$,%
\footnote{Because $\Supp(p_\Omega) = \Omega$,
this is equivalent to the probability-free: for every $\omega \in \Omega$, the configuration $C, \Sigma(\omega)$ is fault-free. However, the probabilistic formulation is preferable, since it is independent of the design choice that $\Supp(p_\Omega) = \Omega$.}
and is 
\emph{terminating} if $\Prob[\text{$C, \Sigma$ is terminating}] = 1$.
\begin{align*}
\Omega_{C,\Sigma} ~ := ~ & \{ (\omega, \pi) \mid 
\text{$\omega \in \Omega$ and $\pi$ is a terminal path from $C,\Sigma(\omega)$}\} \\
p_{\Omega_{C,\Sigma}} (\omega,\pi) ~ := ~ & p_\Omega(\omega) \cdot \Prob(\pi)\, .
\end{align*}
Note that it is because of the condition that $C,\Sigma$ is terminating that $p_{\Omega_{C,\Sigma}}$ is a probability distribution on $\Omega_{C,\Sigma}$. Also note that, for the same reason, the projection function $q_{C,\Sigma} : \Omega_{C,\Sigma} \Epi \Omega$, defined by
$q_{C,\Sigma}(\omega, \pi) := \omega$, is a morphism of sample spaces.
Define $\Tau_{C,\Sigma} : \Omega_{C,\Sigma} \to \State$ by 
\[
\Tau_{C,\Sigma} (\omega,\pi) ~ := ~ \Final(\pi) \, .
\]
It is $\Tau_{C,\Sigma}$ that provides the desired termination random state of the random configuration
$C, \Sigma$.

\begin{definition}[Partial correctness]
\label{definition:partial-correctness}
A specification $\{\Phi\} \, C \, \{\Psi\}$ is \emph{partially correct} if, for every random state $\Sigma: \Omega \to \State$ for which $\Sigma \models \Phi$ we have:
\begin{itemize} 
\item $C, \Sigma$ is fault-free (the \emph{safety guarantee}), and 
\item if $C, \Sigma$ is terminating then $\Tau_{C,\Sigma} \models \Psi$.
\end{itemize}
\end{definition}

\begin{definition}[Total correctness]
A specification $\{\Phi\} \, C \, \{\Psi\}$ is \emph{totally correct} if, for every random state $\Sigma: \Omega \to \State$ for which $\Sigma \models \Phi$ we have:
\begin{itemize} 
\item  $C, \Sigma$ is terminating and $\Tau_{C,\Sigma} \models \Psi$.
\end{itemize}
\end{definition}
\noindent
It is trivial that every totally correct specification is partially correct. 

\section{The frame rule}
\label{section:frame}

\begin{figure}
\[
\begin{prooftree}
\{\Phi\} \, C\, \{\Psi\}
\justifies
\,\{\Phi * \Theta \} \, C\, \{\Psi * \Theta\}\,
\using{\,\FV(\Theta) \cap \MV(C) = \emptyset}
\end{prooftree}
\]
\caption{The frame rule}
\label{figure:frame-rule}
\begin{align*}
\MV(\pAssgn{X}{D}) ~ & := ~ \{X\}  &  \MV(C_1;C_2) ~ & := ~ \MV(C_1) \cup \MV(C_2) \\
\MV(X := E) ~ & := ~ \{X\} & \MV(\If\ B\ \Then\ C_1\ \Else\ C_2) ~ & := ~  \MV(C_1) \cup \MV(C_2) \\
\MV(\Skip)~ & := ~  \emptyset & \MV(\While\ B\ \Do\ C) ~ & := ~  \MV(C)\, .
\end{align*}
\caption{Modified variables $\MV(C)$}
\label{figure:MV}
\end{figure}

The frame rule is given in Figure~\ref{figure:frame-rule}. 
Its formulation is identical to the 
original frame rule formulation in heap separation logic~\cite{YOH}.
The single side condition involves the set 
$\MV(C)$ of all variables that the program $C$ is able to modify, which is defined (in the obvious way) in Figure~\ref{figure:MV}. The lemmas below simply state that variables outside $\MV(C)$ do not change during the execution of $C$. The statement we need later is given by Lemma~\ref{lemma:MV-random}, which asserts this for random state. This is in turn a simple consequence of Lemma~\ref{lemma:MV-deterministic}, which states the property for deterministic state. The straightforward proofs are omitted.
\begin{lemma}
\label{lemma:MV-deterministic}
If $V \subseteq \Var$ is such that $V \cap \MV(C) = \emptyset$, and if $C, \sigma \dTrans^* \tau$ then
$\RState{\tau}{V} = \RState{\sigma}{V}\,$.
\end{lemma}
\begin{lemma}
\label{lemma:MV-random}
If $V \subseteq \Var$ is such that $V \cap \MV(C) = \emptyset$, and if $C,\Sigma$ is a terminating random configuration then $\RState{\Tau_{C,\Sigma}}{V} = \RState{\Sigma}{V} \circ q_{C,\Sigma}\,$.
\end{lemma}
\noindent
Rather than using the syntactic definition of $\MV(C)$  in Figure~\ref{figure:MV}, we remark that one could take the satisfaction of
Lemma~\ref{lemma:MV-deterministic} as a more permissive semantic definition of the set $\MV(C)$ , and the proof of the frame rule below remains valid. 

The main work in proving the soundness of the frame rule lies in proving the theorem below, which
formulates the \emph{relative tightness} property outlined in Section~\ref{section:introduction}. 
This
 asserts a fundamental property of partial (and hence also total) correctness specifications:  the portion of the final state 
$\Tau_{C,\Sigma}$ that is relevant to interpreting the postcondition $\Psi$, 
namely $\RState{\Tau_{C,\Sigma}}{\FV(\Psi)}$, depends only on that part of the start state
$\Sigma$ that is relevant to interpreting the precondition $\Phi$, namely $\RState{\Sigma}{\FV(\Phi)}$. Since the state is randomised, the property of depending only on $\RState{\Sigma}{\FV(\Phi)}$ can be expressed by saying that $\RState{\Tau_{C,\Sigma}}{\FV(\Psi)}$ is  conditionally independent of the input state $\Sigma$ conditional on $\RState{\Sigma}{\FV(\Phi)}$.
In order to formulate this property precisely, the random variables involving $\Sigma$ need to be composed with the
sample space morphism $q_{C,\Sigma} : \Omega_{C,\Sigma} \Epi \Omega$, so that all random variables 
lie over the same sample space $\Omega_{C,\Sigma}$.

\begin{theorem}[Relative tightness]
\label{theorem:independence}
Suppose $\{\Phi\} C \{\Psi\}$ is partially correct, $\Sigma \models \Phi$ and the random configuration 
$C, \Sigma$ is terminating,
then 
\begin{equation}
\label{eqn:relative-tightness}
 \RState{\Tau_{C,\Sigma}}{\,\FV(\Psi)} ~  \Indep ~  \Sigma \circ q_{C,\Sigma} ~ \mid ~
  \RState{\Sigma}{\FV(\Phi)} \circ \, q_{C,\Sigma} ~ . 
\end{equation}
\end{theorem}

\noindent
Before giving the somewhat involved proof of Theorem~\ref{theorem:independence}, we 
show how the soundness of the frame rule follows as a consequence of relative tightness.
\begin{corollary}[Soundness of the frame rule]
\label{corollary:frame}
The frame rule is sound for both partial correctness and total correctness.
\end{corollary}
\begin{proof} We first consider partial correctness. Suppose $\{\Phi\} C \{\Psi\}$ is partially correct. Suppose
also that $\Theta$ is such that $\FV(\Theta) \cap \MV(C) = \emptyset$. 
We need to show that the specification $\{\Phi * \Theta \} C \{\Psi * \Theta\}$ is partially correct.

Suppose then that $\Sigma \models \Phi * \Theta$. That is, we have
\begin{gather}
\label{frame-sound-phi}
\Sigma \models \Phi
\\
\label{frame-sound-theta}
\Sigma \models \Theta
\\
\label{frame-sound-PIT}
\Sigma_{\FV(\Phi)} \Indep \Sigma_{\FV(\Theta)}
\end{gather}
By~\eqref{frame-sound-phi} and the partial correctness of $\{\Phi\} C \{\Psi\}$,
it is immediate that $C, \Sigma$ is fault-free. Assume that $C, \Sigma$ is terminating. We need to show that 
$\Tau_{C,\Sigma} \models \Psi * \Theta$.

 By~\eqref{frame-sound-phi} and the partial correctness of $\{\Phi\} C \{\Psi\}$, we 
have $\Tau_{C,\Sigma} \models \Psi$. 

Since $\FV(\Theta) \cap \MV(C) = \emptyset$,
we have, by Lemma~\ref{lemma:MV-random}, that 
\begin{equation}
\label{frame-sound-equality}
\RState{\Tau_{C,\Sigma}}{\FV(\Theta)} = \RState{\Sigma}{\FV(\Theta)} \! \circ \, q_{C,\Sigma}\, .
\end{equation}
From~\eqref{frame-sound-theta}, it follows that  
$\RState{\Sigma}{\FV(\Theta)} \models \Theta$, by~\eqref{eqn:SAA-SAB}, hence 
$\RState{\Sigma}{\FV(\Theta)} \circ q_{C,\Sigma} \models \Theta$ by (SA3), i.e.,
$\RState{\Tau_{C,\Sigma}}{\FV(\Theta)} \models \Theta$. It then follows from (SA1)
that $\Tau_{C,\Sigma} \models \Theta$.

Finally, we show that $\RState{\Tau_{C,\Sigma}}{\FV(\Psi)} \Indep \RState{\Tau_{C,\Sigma}}{\FV(\Theta)}$.
By~\eqref{frame-sound-PIT}, we have 
 $\RState{\Sigma}{\FV(\Phi)} \circ q_{C,\Sigma} \Indep \RState{\Sigma}{\FV(\Theta)} \circ q_{C,\Sigma}$; equivalently, by~\eqref{frame-sound-equality} (and the symmetry of independence)
\begin{equation}
\label{frame-sound-indep-a}
\RState{\Tau_{C,\Sigma}}{\FV(\Theta)} 
\Indep 
\RState{\Sigma}{\FV(\Phi)} \circ \, q_{C,\Sigma} \, .
\end{equation}
  By Theorem~\ref{theorem:independence}, we have 
$\RState{\Tau_{C,\Sigma}}{\FV(\Psi)}   \Indep  \Sigma \circ q_{C,\Sigma}  \mid 
  \RState{\Sigma}{\FV(\Phi)}  \circ \, q_{C,\Sigma}$.
  It thus follows from~\eqref{cond-indep-def} and~\eqref{indep-function}, using the function $\sigma \mapsto \RState{\sigma}{\FV(\Theta)} : \State \to \State$, that
  $\RState{\Tau_{C,\Sigma}}{\FV(\Psi)}   \Indep  \RState{\Sigma}{\FV(\Theta)} \circ q_{C,\Sigma}  \mid 
  \RState{\Sigma}{\FV(\Phi)}  \circ \, q_{C,\Sigma}\,$; i.e., by~\eqref{frame-sound-equality},
  \begin{equation}
  \label{frame-sound-indep-b}
  \RState{\Tau_{C,\Sigma}}{\FV(\Psi)}   \Indep  \RState{\Tau_{C,\Sigma}}{\FV(\Theta)} 
   \mid 
  \RState{\Sigma}{\FV(\Phi)}  \circ q_{C,\Sigma} \, .
  \end{equation}
  Applying~\eqref{cond-indep-rule} to~\eqref{frame-sound-indep-a} and~\eqref{frame-sound-indep-b}, we get
  \[\RState{\Tau_{C,\Sigma}}{\FV(\Psi)}  \, \Indep \, (\RState{\Tau_{C,\Sigma}}{\FV(\Theta)} , \,
  \RState{\Sigma}{\FV(\Phi)} \! \circ q_{C,\Sigma})\, .\]
  Finally, by another application of~\eqref{indep-function}, using first projection $(\tau,\sigma) \mapsto \tau$ as the function, we obtain
  $\RState{\Tau_{C,\Sigma}}{\FV(\Psi)}   \Indep  \RState{\Tau_{C,\Sigma}}{\FV(\Theta)}$, as required.
  
  The soundness of the frame rule for total correctness follows directly from soundness for partial correctness, because the  termination of $C, \Sigma$, which is the only additional property needed to
to show the total correctness of
 $ \{\Phi * \Theta \} \, C\, \{\Psi * \Theta\}$, is immediate from the assumption that
 $ \{\Phi\} \, C\, \{\Psi \}$ is totally correct.
\end{proof}

We now turn to the proof of Theorem~\ref{theorem:independence}. In preparation, we need a simple lemma about the operational semantics. Given $\sigma,\sigma' \in \State$, we define $\Mask{\sigma}{\sigma'} \in \State$ (the
\emph{masking of $\sigma$ by
$\sigma'$}) by:
\[
\Mask{\sigma}{\sigma'}(X) ~ := ~ \begin{cases}
\sigma'(X) & \text{if $X \in \Dom(\sigma')$} \\
\sigma(X) & \text{otherwise}\, .
\end{cases}
\]
It is immediate from the definition that 
$\sigma' \sqleq \Mask{\sigma}{\sigma'}$. Also we have:
\begin{equation}
\label{equation:masking}
\sigma' \sqleq \sigma ~ \implies ~ \Mask{\sigma}{\sigma'} = \sigma \, .
\end{equation}
\begin{lemma}[Masking lemma]
\label{lemma:masking}
For any state $\sigma$, configuration $C_0,\sigma_0$ and terminal path $\pi \in \TP_{C_0,\sigma_0}$ of the form
\eqref{eqn:terminal-path}, we have $\Mask{\sigma}{\pi} \in \TP_{C_0,\Mask{\sigma}{\sigma_0}}$, where the masked path $\Mask{\sigma}{\pi}$ is defined by
\[
\Mask{\sigma}{\pi} ~ := ~ C_0,\Mask{\sigma}{\sigma_0} \lTrans{L_0} C_1,\Mask{\sigma}{\sigma_1}  \lTrans{L_1}  \cdots\ C_{n-1},
\Mask{\sigma}{\sigma_{n-1}} \lTrans{L_{n-1}} \Mask{\sigma}{\tau}\, ,
\]
Similarly, for every infinite path $\zeta \in \TP_{C_0,\sigma_0}$ of the form
\eqref{eqn:infinite-path}, we have $\Mask{\sigma}{\zeta} \in \IP_{C_0,\Mask{\sigma}{\sigma_0}}$,
where the masked path $\Mask{\sigma}{\zeta}$ is defined by
\[
\Mask{\sigma}{\zeta} ~ := ~
C_0,\Mask{\sigma}{\sigma_0} \lTrans{L_0} C_1,\Mask{\sigma}{\sigma_1}  \lTrans{L_1}  C_{2},
\Mask{\sigma}{\sigma_{2}} \lTrans{L_{2}} \cdots \, ,
\]
\end{lemma}
\noindent
Although we omit the straightforward proof, we remark that
the lemma holds because the inequality $\sigma' \sqleq \Mask{\sigma}{\sigma'}$ means that
transitions of the form $C', \sigma'  \lTrans{L} C'', \sigma''$
and $C', \sigma'  \lTrans{L} \tau'$ are preserved by masking; that is, they give rise to transitions $C', \Mask{\sigma}{\sigma'}  \lTrans{L} C'', \Mask{\sigma}{\sigma''}$
and $C', \Mask{\sigma}{\sigma'}  \lTrans{L} \Mask{\sigma}{\tau'}$ respectively.
The same does not apply to error transitions $C', \sigma'  \dTrans \Fail$, because
the masked state $\sigma$ may be defined on the variable whose undefinedness in
$\sigma'$ triggers the fault.

\begin{proof*}{Proof of Theorem~\ref{theorem:independence}.}
Suppose $\{\Phi\} C \{\Psi\}$ is partially correct, $\Sigma \models \Phi$ and  
$C, \Sigma$ is terminating. 
We have to show that
\[
\RState{\Tau_{C,\Sigma}}{\FV(\Psi)} 
~ \Indep~
\Sigma \circ q_{C,\Sigma} 
    ~ \mid ~
  \RState{\Sigma}{\FV(\Phi)}  \circ \, q_{C,\Sigma}\, .
 \]
 For notational simplicity, define
 $T := \RState{\Tau_{C,\Sigma}}{\FV(\Psi)}$,
 $~S := \Sigma \circ q_{C,\Sigma}$ and
 $R := \RState{\Sigma}{\FV(\Phi)}  \circ q_{C,\Sigma}$.
 Combining~\eqref{cond-indep-def} and~\eqref{independence-via-conditioning},
 we need to show that, for every $\rho \in \Supp(R)$ and every
 $\sigma, \sigma' \in \Supp({S_{|R=\rho}})$, it holds that
 $p_{T|(S|R=\rho)=\sigma} = p_{T|(S|R=\rho)=\sigma'}\,$.
 By a straightforward property of iterated conditioning, $p_{T|(S|R=\rho)=\sigma} = p_{T|(S,R)=(\sigma,\rho)}$.
 Also, by definition of $R$ and $S$, for all $\sigma,\rho \in \State$,
 \[\rho \in \Supp(R) \wedge \sigma \in \Supp({S_{|R= \rho}}) ~\iff ~
 \sigma \in \Supp({S}) \wedge \rho = \RState{\sigma}{\FV(\Phi)} \, .
 \]
 So $\sigma$ determines $\rho$, and hence $p_{T|(S,R)=(\sigma,\rho)} = p_{T|(S,R)=(\sigma,\, \RState{\sigma}{\!\FV(\Phi)})} =
 p_{T|S=\sigma}$.
 Since $\Supp(S) = \Supp(\Sigma)$, it is therefore enough to show:
  \begin{equation}
  \label{main-to-show}
 \forall\sigma,\sigma' \in \Supp({\Sigma}), ~~\RState{\sigma}{\FV(\Phi)} =  \RState{\sigma'}{\FV(\Phi)}
 ~ \implies ~ p_{T|S=\sigma} = p_{T|S=\sigma'} \, .
 \end{equation}
 We approach this in several steps.
 
 
 Consider any $\sigma \in \Supp(\Sigma)$.
 Our first goal is to show that $C, \RState{\sigma\!}{\FV(\Phi)}$ is terminating.
  By Lemma~\ref{lemma:masking}, every infinite execution path $\zeta \in \IP_{C,\,\RState{\sigma}{\FV(\Phi)}}$, gives us a corresponding infinite path $\Mask{\sigma}{\zeta} \in \IP_{C,\Mask{\sigma}{(\RState{\sigma}{\FV(\Phi)})}}$. 
Then  $\Mask{\sigma}{\zeta} \in \IP_{C,\sigma}$, because $\Mask{\sigma}{(\RState{\sigma}{\FV(\Phi)})} = \sigma$, by~\eqref{equation:masking}. That is, every path in $\IP_{C,\,\RState{\sigma\,}{\FV(\Phi)}}$ has a matching path in 
$\IP_{C,\sigma}$. So, by the formula for divergence probabilities~\eqref{eqn:divergence-prob}, we have
\[
\Prob[\text{$C,{\RState{\sigma\!}{\FV(\Phi)}}$ diverges}]  ~ \leq ~ 
\Prob[\text{$C,\sigma$ diverges}] \, .
\]
Since $C,\Sigma$ is terminating and  $\sigma \in \Supp(\Sigma)$, it holds that
$C,\sigma$ is terminating, hence
$\Prob[\text{$C,\sigma$ diverges}] = 0$. Therefore $\Prob[\text{$C,{\RState{\sigma\!}{\FV(\Phi)}}$ diverges}] = 0$.
Because $\Sigma \models \Phi$, it holds that $\RState{\Sigma}{\FV(\Phi)} \models \Phi$.
 So, by the safety guarantee from the partial correctness of $\{\Phi\} C \{\Psi\}$, the random configuration
 $C, \RState{\Sigma}{\FV(\Phi)}$ is fault-free. Since 
 $\RState{\sigma\!}{\FV(\Phi)} \in \Supp({\RState{\Sigma}{\FV(\Phi)}})$, 
 it holds that $C, \RState{\sigma}{\FV(\Phi)}$ is fault-free.
 Hence, 
 by~\eqref{terminate-inequality} and the discussion following it,
 $\Prob[\text{$C,{\RState{\sigma\!}{\FV(\Phi)}}$ terminates}] = 1$. 
 That is, $C,{\RState{\sigma\!}{\FV(\Phi)}}$ is terminating. 
 
 Since both $C,\sigma$  and $C, \RState{\sigma\!}{\FV(\Phi)}$ are terminating, they induce
probability distributions $p_{C,\sigma}$ and $p_{C, \,\RState{\sigma}{\FV(\Phi)}}$ on their terminal states.
Our next goal is to relate these two distributions.
By Lemma~\ref{lemma:masking}, every terminal path $\pi \in \TP_{C, \,\RState{\sigma}{\FV(\Phi)}}$,
induces a corresponding terminal path $\Mask{\sigma}{\pi} \in \TP_{C,\sigma}$
with $\Prob(\Mask{\sigma}{\pi}) = \Prob(\pi)$. Since $C, \RState{\sigma\!}{\FV(\Phi)}$ is terminating,
the probabilities $\Prob(\pi)$ of all $\pi \in \TP_{C, \,\RState{\sigma}{\FV(\Phi)}}$ add up to $1$,
hence so do all probabilities $\Prob(\Mask{\sigma}{\pi})$. Thus  every terminal path
in $\TP_{C,\sigma}$ is necessarily of the form $\Mask{\sigma}{\pi}$ for some 
$\pi \in \TP_{C, \,\RState{\sigma}{\FV(\Phi)}}$.
 It follows that the induced probabilities on terminal states, 
$p_{C,\sigma}$ and $p_{C, \,\RState{\sigma}{\FV(\Phi)}}$, are related to each other
by the pushforward property
\begin{equation}
\label{pushforward-property}
p_{C,\sigma} ~ = ~  (\tau \mapsto \Mask{\sigma}{\tau})_! \, (p_{C, \RState{\sigma}{\FV(\Phi)}}) \, .
\end{equation}

 We next apply the partial correctness
 of $\{\Phi\} C \{\Psi\}$ to the random state $\RState{\Sigma}{\FV(\Phi)}$.
 We have already noted that $\RState{\Sigma}{\FV(\Phi)} \models \Phi$.
 We have also shown that $C, \RState{\sigma\!}{\FV(\Phi)}$ is  terminating
 for every $\sigma \in \Supp(\Sigma)$. 
 Since every $\sigma' \in \Supp({\RState{\Sigma}{\FV(\Phi)}})$ is of the form 
 $\RState{\sigma\!}{\FV(\Phi)}$ for $\sigma \in \Supp(\Sigma)$, it follows that
 the random  configuration $C, \RState{\Sigma\!}{\FV(\Phi)}$ is terminating.
 So 
$\Tau_{C,\, \RState{\Sigma\,}{\FV(\Phi)}} \models \Psi$ holds, by the partial correctness
 of $\{\Phi\} C \{\Psi\}$.
It follows that the random state $\Tau_{C,\, \RState{\Sigma\,}{\FV(\Phi)}}$ is $\FV(\Psi)$-total.

We now finally turn to~\eqref{main-to-show}. What we actually prove is
\begin{equation}
\label{actual-shown}
\forall\sigma\in \Supp({\Sigma}), ~~
p_{T|S=\sigma} ~ =  ~
(\tau' \mapsto \RState{\tau'}{\FV(\Psi)})_!  \, (p_{C, \,\RState{\sigma\,}{\FV(\Phi)}}) \, ,
\end{equation}
from which~\eqref{main-to-show} follows, since the right-hand side of~\eqref{actual-shown} depends only on $\RState{\sigma}{\FV(\Phi)}$.
 To establish~\eqref{actual-shown}, consider any $\sigma \in   \Supp(\Sigma)$.
 We have $T = \RState{\Tau_{C,\Sigma}}{\FV(\Psi)} = (\tau' \mapsto \RState{\tau'}{\FV(\Psi)}) \circ \Tau_{C,\Sigma}$. Also, by the definition of $p_{C,\sigma}$, it holds that
 $p_{\,\Tau_{C,\Sigma}| S = \sigma} = p_{C,\sigma}$.
 We can thus establish~\eqref{actual-shown} by
 \begin{align*}
 p_{T|S=\sigma} ~ & = ~ p_{\,((\tau' \mapsto \RState{\tau'}{\FV(\Psi)}) \circ \Tau_{C,\Sigma}) | S=\sigma}
 \\
& = ~ (\tau' \mapsto \RState{\tau'}{\FV(\Psi)})_!\, (p_{\,\Tau_{C,\Sigma} | S=\sigma})
 \\
 & = ~  (\tau' \mapsto \RState{\tau'}{\FV(\Psi)})_!\, (p_{C,\sigma}) 
 \\
 & = ~ (\tau' \mapsto \RState{\tau'}{\FV(\Psi)})_!\,
  (\tau \mapsto \Mask{\sigma}{\tau})_! \, (p_{C, \RState{\sigma}{\FV(\Phi)}}) 
  & & \text{by~\eqref{pushforward-property}}
\\
& = ~ (\tau \mapsto \RState{\tau}{\FV(\Psi)})_!\,
 (p_{C, \RState{\sigma}{\FV(\Phi)}}) \,,
 \end{align*}
where the last step uses the fact, established earlier, that 
$\Tau_{C,\, \RState{\Sigma\,}{\FV(\Phi)}}$ is $\FV(\Psi)$-total.
Indeed, this means that every $\tau \in \Supp(p_{\,C, \RState{\sigma}{\FV(\Phi)}})$ is 
$\FV(\Psi)$-total. Hence,
$\RState{\tau}{\FV(\Psi)} =  \RState{(\Mask{\sigma}{\tau})}{\FV(\Psi)}$ holds for such $\tau$, since the masking operation does not change the value of variables defined in $\tau$. 
\end{proof*}

We end this section with a very simple example illustrating how the safety aspect of partial correctness is essential to the validity of both relative tightness and the frame rule. Consider the specification
\begin{equation}
\label{violation}
\{\top \} ~ X :=X\ \mathsf{mod}\ 2 ~ \{[X = 0 \vee X=1]\} \, ,
\end{equation}
where $\top$ is the  true assertion. 
Recall that $[X = 0 \vee X=1]$ says that the stated property holds with probability $1$.
Under our definition of partial correctness (Definition~\ref{definition:partial-correctness}), this specification is not partially correct because it fails the safety guarantee: the empty random state satisfies the precondition, 
but execution from the empty state faults.  If, however, the safety guarantee is 
dropped from the definition of partial correctness, then the above specification becomes correct.
In any random state $\Sigma$ from which the program $X {:=} X\ \mathsf{mod}\ 2$ terminates,
the postcondition $\{[X = 0 \vee X=1]\}$ is indeed true.

We observe that considering~\eqref{violation} as correct violates relative tightness.
Let $\Sigma$ be a  random state 
with distribution (expressed as a convex sum of outputs weighted by probabilities) 
\[
p_\Sigma ~ = ~ \frac{1}{2}\cdot [X\mapsto 0] + \frac{1}{2}\cdot [X\mapsto 1] \, .
\]
Obviously $\Sigma$ satisfies the precondition $\top$. Moreover $X {:=} X\ \mathsf{mod}\ 2$ terminates,
when run from $\Sigma$, in 
the random state
\[
\Tau ~ = ~ \Sigma \circ q_{\,{X {:=} X\,\mathsf{mod}\,2},\, \Sigma} \, ,
\]
whose distribution is 
\[
p_\Tau ~ = ~ \frac{1}{2}\cdot [X\mapsto 0] + \frac{1}{2}\cdot [X\mapsto 1] \, .
\]
Since  $\FV(\top) = \emptyset$, and $\FV([X = 0 \vee X=1]) = \{X\}$, 
the property of relative tightness~\eqref{eqn:relative-tightness}
asserts that $\Tau \Indep  \Sigma \circ q_{\,{X {:=} X\,\mathsf{mod}\,2},\, \Sigma}$, 
i.e., $\Tau \Indep \Tau$, which is patently not the case. 

For a similar reason, the frame rule also fails if~\eqref{violation} 
is considered as correct. Consider the specification
\begin{equation}
\label{frame-violation}
\{\top * [Y = 0 \vee Y=1]\} ~ X :=X\ \mathsf{mod}\ 2 ~ \{[X = 0 \vee X=1] * [Y = 0 \vee Y=1]\} \, ,
\end{equation}
Let $\Sigma$ be a random state with distribution
\[
p_\Sigma ~ := ~ \frac{1}{2}\cdot [X\mapsto 0, Y \mapsto 0] + \frac{1}{2}\cdot [X\mapsto 1, Y \mapsto 1] \,
\]
Then $\Sigma$ satisfies the precondition $\top * [Y = 0 \vee Y = 1]$. 
The command $X {:=} X\ \mathsf{mod}\ 2$ again terminates,
when run from $\Sigma$, in a random state $\Tau$ with $p_\Tau = p_\Sigma$.
Because the vaues of $X$ and $Y$ are correlated in the distribution $p_\Tau$, it is not the case that $\Tau \models [X= 0 \vee X = 1] * [Y = 0 \vee Y = 1]$. So~\eqref{frame-violation} does not hold.

Let us use the same example to compare our frame rule (Figure~\ref{figure:frame-rule}) with the frame rule from
the original probabilistic separation logic~\cite{BHL}. 
The translation of specification~\eqref{violation} into the assertion logic of~\cite{BHL} produces a correct
specification according to their interpretation. However, the frame rule from~\cite{BHL} is not 
applicable,  because~\eqref{violation} fails the side-condition that the 
specification's precondition must imply that the starting state is defined on variables whose initial values are read by the program. It is this side-condition that prevents the  incorrect specification~\eqref{frame-violation} from being derived. The frame rule of~\cite{BHL}
 also has another side-condition, which places a restriction on the variables that 
are allowed to appear in the specification's postcondition. Neither of the two side-conditions under discussion appears
 in the frame rule of the present paper (Fig.~\ref{figure:frame-rule}),
because the safety guarantee and the derived property of relative tightness render them unnecessary.
 
 Our proof of soundness for the frame rule depends on the property of relative tightness (Theorem~\ref{theorem:independence}).
An analogue of this property is crucial to the proof of soundness for the frame rule in~\cite{BHL}; specifically, a ``soundness'' property for classes of variables, which says that behaviour of a program $C$  can be described as a probabilistic map from variables whose initial values are read by $C$ to variables that are written to by $C$ (Lemma~6(3) of \emph{op.\ cit.}). The analogy is that this probabilistic map establishes that  a portion of the final state (the write variables) is independent of the input state conditional on the projection of the latter to  the read variables. A major difference, however, is that our relative tightness is a derived property of \emph{specifications}, relating to portions of the input and output states determined by the pre-\ and postcondition respectively; whereas the analogous property in~\cite{BHL} is 
based on a taxonomy of program variables. One other difference is that our results apply to the full pwhile language with general unbounded while loops. This additional generality necessitates the termination argument that appears in our proof of relative tightness.

\section{Further work}
\label{section:proof}


Having established the soundness of the frame rule, the most pressing question is whether the framework introduced in this paper can serve as the basis for a practicable probabilistic separation  logic for verifying interesting probabilistic programs. 
Given our departure from the standard, resource-monoid-based formulation of probabilistic separation logic,
and our relaxation of the restrictions on imperative programs that are built into the design of the proof systems in~\cite{BHL,BGHT} 
(the special treatment of deterministic variables, and the restrictions to deterministic loop guards and to bounded loops),
there is no \emph{a priori} reason to expect that a reasonable such proof system exists.  Nevertheless, we are optimistic that one does.

As a first step in this direction, in his undergraduate project work~\cite{jereb}, the
first author has checked that the example verifications of cryptographic protocols from~\cite{BHL} do transfer to our
semantic framework. This has been carried out in a setting in which the
aforementioned restrictions on programs have been retained,  using
proof rules for partial correctness closely modelled on those in \emph{op.\ cit.}

\begin{figure}
\begin{gather*}
\begin{prooftree}
\{\Phi\} \, \If\ B\ \Then\ C\ \Else\ \Skip\, \{\Phi\}
\justifies
\{\Phi\} \, \While\ B\ \Do\ C\, \{[\neg B] \wedge \Phi\}
\using{\,\text{$\Phi$ topologically closed}}
\end{prooftree}
\end{gather*}
\caption{A provisional proof rule for general while loops}
\label{figure:proof-rules}
\end{figure}

Finding a more general proof system, catering for unrestricted programs, is an interesting 
research goal. We believe that our specification framework provides a promising basis for the development of such a system. 
For example, 
Figure~\ref{figure:proof-rules} presents a plausible 
 partial-correctness proof rule for while loops. It has a side condition that the loop invariant $\Phi$ must be \emph{topologically closed},
meaning that, for every sample space $\Omega$, the set 
$\{ \Sigma : \Omega \to \State \mid \Sigma \models \Phi\}$ is closed
in the topology of convergence-in-probability of random variables (which, because we are considering random variables valued in a discrete space $\State$, coincides with almost-sure convergence).
Regarding the full system we have in mind, we comment that we envisage a crucial interaction between conditioning modalities $\CM{E} \,\Phi$ in  the logic  and conditionals in the programming language,
and also a conditional generalisation of the frame rule, allowing  frame-rule-like inferences to be applied to conditional independence statements.
Such a system is the subject of current investigation.

One final  point for  investigation is  the extent to which our classical logic, based on a partial satisfaction relation, forms a
sufficiently expressive alternative to probabilistic separation logic based on partial resource monoids. For example, does it support a natural (and useful) version of the separating implication (``magic wand'') connective of separation logic?

\begin{ack}
We thank the anonymous MFPS reviewers for their helpful comments.
\end{ack}

\bibliographystyle{./entics}

\end{document}